\documentclass[oneside,reqno,english]{amsart}
\usepackage[T1]{fontenc}
\usepackage[latin9]{inputenc}
\usepackage{babel}
\usepackage{amstext}
\usepackage{amsthm}
\usepackage{amssymb}
\usepackage{esint}
\usepackage{url}

\makeatletter
\numberwithin{equation}{section}
\numberwithin{figure}{section}
\theoremstyle{plain}
\newtheorem{thm}{\protect\theoremname}
\theoremstyle{definition}
\newtheorem{defn}[thm]{\protect\definitionname}
\theoremstyle{plain}
\newtheorem{prop}[thm]{\protect\propositionname}
\theoremstyle{plain}
\newtheorem{lem}[thm]{\protect\lemmaname}

\providecommand{\definitionname}{Definition}
\providecommand{\lemmaname}{Lemma}
\providecommand{\propositionname}{Proposition}
\providecommand{\theoremname}{Theorem}

\makeatother
\begin{document}
\title{Clifford-wavelet Transform and the uncertainty principle}
\author{Hicham Banouh}

\address{Laboratoire AMNEDP, Facult\'e de Math\'ematiques, Universit\'e de Sciences et Technologie Houari Boumedienne, Bab Zouar, Alger, Algeria.}
\email{banouh\_hicham@hotmail.fr}
\author{Anouar ben mabrouk}
\address{Higher institut of Applied Mathematics and Computer Sciences, University of Kairaouan, Street of Assad Ibn Alfourat, Kairaouan 3100, Tunisia.\newline
and\newline
Laboratory of Algebra, Number Theory and Nonlinear Analysis, Department of Mathematics, Faculty of Sciences, Monastir, Tunisia.\newline
and\newline
Department of Mathematics, Faculty of Sciences, University of Tabuk, King Faisal Rd, Tabuk, Saudi Arabia.}
\email{anouar.benmabrouk@fsm.rnu.tn}
\author{Mohamed Kesri}
\address{D\'epartement d'analyse, Universit\'e de Sciences et Technologie Houari Boumedienne, Bab Zouar, Alger, Algeria.}
\email{mkes33dz@yahoo.fr}
\begin{abstract}
In this paper we derive a Heisenberg-type uncertainty principle for
the continuous Clifford wavelet transform. A brief review of Clifford algebra/analysis, wavelet transform on $\mathbb{R}$ and Clifford-Fourier
transform and their proprieties has been conducted. Next, such concepts have been applied to develop an uncertainty principle based on Clifford wavelets.
\end{abstract}
\maketitle
\section{Introduction}
Transformations such as the Fourier one are powerful methods for signals representations and features detection in signals. The signals are transformed from the original domain to the spectral or frequency one. In the frequency domain many characteristics of a signal are seen more clearly. Contrarily to the Fourier modes, wavelet basis functions are localized in both spatial and frequency domains and thus yield very sparse and well-structured representations of signals, important facts from a signal processing point
of view. The first work on wavelet analysis has been done by Morlet in \cite{Morlet1982} to study seismic waves. He also, with Grossman, gave a mathematical study of continuous wavelet transform (see \cite{Grossmann1984}). In \cite{Meyer1986}, Meyer recognised the link between harmonic analysis and Morlet's theory and gave a mathematical foundation to the continuous wavelet theory. The continuous-wavelet analysis of a square integrable function $f$
begins by a convolution with copies of a given ``mother wavelet'' $\psi$ translated and dilated respectively by $b\in\mathbb{R}$ and $a>0$. Such a function $\psi$ has to fulfil an admissibility condition which states that
\[
\mathcal{A}_{\psi}=\int_{\mathbb{R}}\frac{\left|\widehat{\psi}(\xi)\right|^{2}}{\left|\xi\right|}d\xi<+\infty,
\]
where $\widehat{\psi}$ is the classical Fourier transform of $\psi$.
More information on real wavelet can be found in \cite{Grossmann1986}
and \cite{Debnath2015} and a generalization to Sobolev spaces $\mathcal{H}^{s}(\mathbb{R})$ for an arbitrary real number $s$ in \cite{Rieder1990}.

On the other hand, Clifford analysis leads to the generalization of
real and harmonic analysis to higher dimensions. Clifford algebra
accurately treats geometric entities depending on their dimension
such as scalars, vectors, bivectors and volume elements, etc. The
distinction of axial and polar vectors in physics, e.g. is resolved
in the form of vectors and bivectors. For example, the quaternion
description of rotations is fully incorporated in the form of rotors.
With respect to the geometric product of vectors, division by non-zero
vectors is defined. Clifford analysis/algebra has started to take
place especially in signal and image processing (See for instance \cite{Brackx2001a}, \cite{Clifford1882}, \cite{DeSchepper2006}, \cite{Dirac1928}, \cite{Hamilton1844}, \cite{Hamilton1866}, \cite{Sommen2015}).

The present paper lies in the same topic of Clifford algebra/analysis
applications. We aim to develop an uncertainty principle proof in
the Clifford analysis framework based on Clifford wavelets.

The paper is organized as follows: In section 2 we give a brief review
of the Clifford analysis, introduce the notion of wavelet transform
in $\mathbb{R}$ and the uncertainty principle. The third section
is devoted to some results and properties of the
Clifford-Fourier transform. In section 4, Clifford-wavelet transform
has been investigated. In section 5, the uncertainty principle for
the Clifford wavelet transform is established.
\section{Preliminaries}
In this section, we aim to recall the basic properties of the Clifford
algebras (See \cite{Hamilton1844}, \cite{Sommen2015} and the references
therein). Next, a brief review of continuous wavelet transform
in $\mathbb{R}$ and the Heisenberg uncertainty principle are developed.
\subsection{Clifford Algebras}
The Clifford algebra $\mathbb{R}_{n}$ associated to $\mathbb{R}^{n}$
is an associative algebra generated by an orthonormal (the canonical)
basis $\left\{ e_{1},e_{2},\dots,e_{n}\right\} $ by means of a non
commutative product
\[
e_{i}e_{j}+e_{j}e_{i}=-2\delta_{ij},
\]
where $\delta$ is the Kronecker symbol. This yields a finite $2^{n}$-dimensional algebra known as the Clifford algebra $\mathbb{R}_{n}$. It is decomposed
as a direct sum
\[
\mathbb{R}_{n}=\mathbb{R}_{n}^{0}\oplus\mathbb{R}_{n}^{1}\oplus\cdots\oplus\mathbb{R}_{n}^{n},
\]
where $\mathbb{R}_{n}^{k}$ are the spaces of $k$-multi vectors defined
by
\[
\mathbb{R}_{n}^{k}=span_{\mathbb{R}}\left\{ e_{A}=e_{i_{1}i_{2}\dots i_{k}},1\leq i_{1}<i_{2}<\cdots<i_{k}\leq n\right\},
\]
where $e_{A}=e_{i_{1}i_{2}\dots i_{k}}=e_{i_{1}}e_{i_{2}}\cdots e_{i_{k}}$. We may also decompose $\mathbb{R}_{n}$ as two sub-algebras
\begin{align*}
	\mathbb{R}_{n}=\mathbb{R}_{n}^{+}\oplus\mathbb{R}_{n}^{-}=\bigoplus_{k\text{even}}\mathbb{R}_{n}^{k}\oplus\bigoplus_{k\text{odd}}
	\mathbb{R}_{n}^{k}
\end{align*}
called respectively the even and odd sub-algebras. Consequently, any Clifford number $a\in\mathbb{R}_{n}$ has a representation
of the form
\[
a=\sum_{A}a_{A}e_{A},\;a_{A}\in\mathbb{R},
\]
where $e_{A}=e_{i_{1}i_{2}\cdots i_{k}}=e_{i_{1}}e_{i_{2}}\cdots e_{i_{k}}$,
with $1\leq i_{1}<i_{2}<\cdots<i_{k}\leq n$, $e_{\emptyset}=1$ and $A=(i_1,i_2,\dots,i_k)$. Denoting $|A|$ the length or the cardinality for the multi-index $A$, the element $a$ may be written as
\[
a=\sum_{k=0}^{n}\sum_{\left|A\right|=k}a_{A}e_{A},
\]
On the algebra $\mathbb{R}_{n}$, we may introduce some involutive
operators such as
\begin{itemize}
	\item \textbf{Main-involution}: $\widetilde{e_{j}}=-e_{j}$, $\forall j$,
	which yiels that $\widetilde{e_{A}}=(-1)^{\left|A\right|}e_{A}$ and
	consequently, $\widetilde{ab}=\widetilde{a}\widetilde{b}$, $\forall a,b\in\mathbb{R}_{n}$.
	\item \textbf{Reversion}: $e_{j}^{\ast}=e_{j}$, $\forall j$, which in turns
	yields that $e_{A}^{\ast}=(-1)^{\frac{\left|A\right|(\left|A\right|-1)}{2}}e_{A}$
	and thus $(ab)^{\ast}=b^{\ast}a^{\ast}$, $\forall a,b\in\mathbb{R}_{n}$.
	\item \textbf{Conjugation:} $\overline{e_{j}}=-e_{j}$, $\forall j$, yielding
	that $\overline{e_{A}}=(-1)^{\frac{|A|(|A|+1)}{2}}e_{A}$ and consequently,
	$\overline{ab}=\overline{b}\overline{a}$, $\forall a,b\in\mathbb{R}_{n}$.
\end{itemize}
The concept of real Clifford algebra can be extended to the complex
Clifford algebra $\mathbb{C}_{n}=\mathbb{R}_{n}+i\mathbb{R}_{n}$.
An element $\lambda\in\mathbb{C}_{n}$ may be written on the form $\lambda=\sum_{A}\lambda_{A}e_{A},\,\lambda_{A}\in\mathbb{C}$
and thus possesses the decomposition $\lambda=a+ib$, $a,b\in\mathbb{R}_{n}$.
This induces the
\begin{itemize}
	\item \textbf{Hermitian conjugation} $\lambda^{\dagger}=\overline{a}-i\overline{b}$.
\end{itemize}
In this context, a vector $x=(x_{1},x_{2},\dots,x_{n})\in\mathbb{R}^{n}$
may be identified to the Clifford element in $\mathbb{R}_{n}$, $\underline{x}={\displaystyle \sum_{j=1}^{n}x_{j}e_{j}}$.
This permits to define the Clifford product of two vectors by
$$
\underline{x}\underline{y}=\underline{x}\cdot\underline{y}+\underline{x}\wedge\underline{y},
$$
where the $\cdot$ product is
$$
\underline{x}\cdot\underline{y}=-<\underline{x},\underline{y}>=-{\displaystyle \sum_{j=1}^{n}x_{j}y_{j}},
$$
the classical inner product on $\mathbb{R}^{n}$ and where the $\wedge$
product is the outer product $\underline{x}\wedge\underline{y}={\displaystyle \sum_{j<k}e_{j}e_{k}(x_{j}y_{k}-x_{k}y_{j})}$. This yields that
\[
\underline{x}\cdot\underline{y}=\frac{1}{2}(\underline{x}\underline{y}+\underline{y}\underline{x})\;\;\hbox{and}\;\;
\underline{x}\wedge\underline{y}=\frac{1}{2}(\underline{x}\underline{y}-\underline{y}\underline{x}).
\]
In particular we have
\[
\underline{x}^{2}=-\left|\underline{x}\right|^{2}=-\sum_{j=1}^{n}|x_{j}|^{2}.
\]
Any vector $\underline{x}$ is decomposed as $\underline{x}=\underline{x}_{\parallel\underline{\omega}}+\underline{x}_{\perp\underline{\omega}}$
for a $\underline{\omega}\in\mathcal{S}^{n-1}$ with $\underline{\omega}\cdot\underline{x}_{\perp\underline{\omega}}=0$
and $\underline{\omega}\wedge\underline{x}_{\parallel\underline{\omega}}=0$,
which in turns induces that $\underline{x}_{\parallel\underline{\omega}}=\left\langle \underline{x},\underline{\omega}\right\rangle \underline{\omega}$
and $\underline{x}_{\perp\underline{\omega}}=\underline{\omega}\left(\underline{x}\wedge\underline{\omega}\right)$.
This permits next to characterize the reflection $R_{\underline{\omega}}$
with respect to the hyperplane $\underline{\omega}^{\perp}$ as $R_{\underline{\omega}}(\underline{x})=\underline{\omega}\underline{x}\underline{\omega}$. Cartan-Dieudonné Theorem (\cite{Cartan1966}) relates the reflection
to the so-called spinors. It states that there exists elements $\underline{\omega}_{1},\underline{\omega}_{2},\dots,\underline{\omega}_{2l}\in\mathcal{S}^{n-1}$ with $\underline{\omega}_{j}^{2}=-1,1\leq j\leq2l$ and a rotation $T\in SO(n)$ such that
\[
T(\underline{x})=\left[R_{\underline{\omega}_{1}}\circ R_{\underline{\omega}_{2}}\circ......\circ R_{\underline{\omega}_{2l}}\right](\underline{x})=\underline{\omega}_{1}\underline{\omega}_{2},......\underline{\omega}_{2l}\underline{x}\,\underline{\omega}_{2l}\underline{\omega}_{2l-1},......\underline{\omega}_{2}\underline{\omega}_{1}.
\]
Denoting next $s=\underline{\omega}_{1}\underline{\omega}_{2},......\underline{\omega}_{2l}$
and $\overline{s}=\underline{\omega}_{2l}\underline{\omega}_{2l-1},......\underline{\omega}_{2}\underline{\omega}_{1}$, we have $T(\underline{x})=s\underline{x}\overline{s}$. The element
$s$ is called a spinor. Generally speaking, the spin group of order $n$ is
\[
Spin(n)=\left\{ s\in\mathbb{R}_{n};\;s={\displaystyle \prod\limits _{j=1}^{2l}}\underline{\omega}_{j},\;\underline{\omega}_{j}^{2}=-1,1\leq j\leq2l\right\}.
\]
To finish with this brief overview on Clifford algebra/analysis,
it remains to recall the functional framework. Let $f:\mathbb{R}^{n}\longrightarrow\mathbb{C}_{n}$. It may be expressed as
\[
f(\underline{x})=\sum_{A}f_{A}(\underline{x})e_{A},
\]
where $f_{A}$ are generally $\mathbb{C}$-valued functions and $A\subset\left\{ 1,2,\cdots,n\right\} $. The inner product of two functions, $f=\displaystyle\sum_{A}f_{A}(\underline{x})e_{A}$
and $g=\displaystyle\sum_{B}g_{B}(\underline{x})e_{B}$ is defined by
\begin{equation}\label{eq:inner product}
<f,g>_{L^{2}(\mathbb{R}_{n},dV(\underline{x}))} =\int_{\mathbb{R}^{n}}\left[f(\underline{x})\right]^{\dagger}g(\underline{x})dV(\underline{x})
\end{equation}
and the associated norm by
$$
\|f\|_{L^{2}(\mathbb{R}_{n},dV(\underline{x}))} =<f,f>_{L^{2}(\mathbb{R}_{n},dV(\underline{x}))}^{\frac{1}{2}}.
$$
We denote also
\[
\left\Vert f\right\Vert _{L^{1}(\mathbb{R}_{n},dV(\underline{x}))}=\int_{\mathbb{R}^{n}}|f(\underline{x})|dV(\underline{x})
\]
where $dV(\underline{x})$ stands for the Lebesgue measure. The inner product (\ref{eq:inner product}) satisfies the Cauchy-Schwartz
inequality
\begin{equation}
\left|<f,g>_{L^{2}(\mathbb{R}_{n},dV(\underline{x}))}\right|\leq\left\Vert f\right\Vert _{L^{2}(\mathbb{R}_{n},dV(\underline{x}))}\left\Vert g\right\Vert _{L^{2}(\mathbb{R}_{n},dV(\underline{x}))}\label{cauchy schwartz}
\end{equation}
\subsection{Wavelets on $\mathbb{R}$}
Wavelet analysis of functions is based as the Fourier one on some
type of transform known as wavelet transform which consists in some
product and/or projection of the function on suitable windows issued
from one source analyzing function called mother wavelet and which
plays the role of the Fourier mode $e^{ix}$. Denote $\psi$ such
function. It should satisfy
\begin{itemize}
	\item A finite energy or space/time localization assumption: $\psi\in L^{2}(\mathbb{R},dx)$.
	\item An admissibility assumption stating that
	\[
	0<\mathcal{A}_{\psi}=\displaystyle \int_{\mathbb{R}}\frac{(\widehat{\psi}(\xi))^{2}}{\left\vert \xi\right\vert }d\xi<\infty,
	\]
	where $\widehat{\psi}(\xi)=\displaystyle\int_{\mathbb{R}}\psi(x)e^{-ix\xi}dx$
	is the classical Fourier transform of $\psi$.
	\item Vanishing moments:
	\[
	\exists N\in\mathbb{N},\int\limits _{\mathbb{R}}x^{n}\psi(x)dx=0,\forall n\leq N.
	\]
\end{itemize}
\begin{defn}
	(Continuous-wavelet transform) Let $f\in{L}^{2}(\mathbb{R},dx)$.
	Its continuous wavelet transform is defined as
	\[
	T_{\psi}[f](a,b)=\int\limits _{\mathbb{R}}\overline{\psi^{a,b}(x)}f(x)dx=(\widetilde{\psi_{a}}\ast f)(b),
	\]
	where $\psi^{a,b}(x)=\frac{1}{\sqrt{a}}\psi(\frac{x-b}{a})$, $(a>0,\,b\in\mathbb{R})$
	and $\widetilde{\psi_{a}}(x)=\frac{1}{\sqrt{a}}\overline{\psi(\frac{-x}{a})}.$
\end{defn}
Using the admissibility assumption above, we immediately affirm that
\[
\int_{\mathbb{R}^{\ast}}\int_{\mathbb{R}}\left\vert T_{\psi}[f](a,b)\right\vert ^{2}db\frac{da}{a^{2}}<\infty.
\]
More precisely, an inner product may be defined for the wavelet transforms
as
\[
\left[T_{\psi}\left[f\right],T_{\psi}\left[g\right]\right]=\frac{1}{\mathcal{A}_{\psi}}\int\limits _{\mathbb{R}}\int\limits _{\mathbb{R}}\overline{T_{\psi}\left[f\right](a,b)}T_{\psi}\left[g\right](a,b)\frac{da}{a^{2}}db,
\]
which in turns may be related to the inner product of the analyzed functions $f$ and $g$ by means of a Fourier-Plancherel type formula
as
\[
\left[T_{\psi}\left[f\right],T_{\psi}\left[g\right]\right]=<f,g>_{{L}^{2}(\mathbb{R},dx)}.
\]
Moreover, a reconstruction formula may also be induced yielding that
\[
f(x)=\frac{1}{\mathcal{A}_{\psi}}\int\limits _{\mathbb{R}}(\int\limits _{\mathbb{R}^{+}}\psi^{a,b}(x)T_{\psi}\left[f\right](a,b)\frac{da}{a^{2}})db.
\]

\subsection{Uncertainty Principle}

The uncertainty principle also known as Heisenberg's uncertainty principle
discovered in 1927 by Heisenberg in \cite{Heisenberg1927} is certainly
one of the most famous and important concepts of quantum mechanics. It plays an important role in the development and understanding of quantum physics. The physical origin of uncertainty principle is related to quantum systems and states that: the determination of positions by performing measurement on the system disturbs it sufficiently to make the determination of momentum imprecise and vice-versa.

Mathematically, in wave mechanics, the uncertainty relation between
position and momentum arises because the expressions of the wave function
in the two corresponding orthonormal bases in Hilbert space are Fourier
transforms of one another (i.e., position and momentum are conjugate
variables). Hence, a non-zero function and its Fourier transform cannot
both be sharply localized. The next theorem formally summarizes the
Heisenberg's Uncertainty Principle
\begin{thm}\label{uncer cas general} (Uncertainty Principle \cite{Weyl1950}) Let $A$ and $B$ be two self-adjoint operators on a Hilbert space $\mathcal{X}$ with domains $\mathcal{D}(A)$ and $\mathcal{D}(B)$ respectively and denote finally $\left[A,B\right]=AB-BA$ their commutator. Then
	\begin{equation}
	\left\Vert Af\right\Vert _{2}\left\Vert Bf\right\Vert _{2}\geq\displaystyle\frac{1}{2}\left\vert <\left[A,B\right]f,f>\right\vert ,\forall f\in\mathcal{D}(\left[A,B\right]).
	\end{equation}
\end{thm}
\section{Clifford-Fourier Transform}
In this section we propose to review some basic concepts of the Clifford-Fourier transform. Fore more details we may refer to \cite{Fu2015} and \cite{Jday2018}. The Clifford-Fourier transform of a Clifford-valued function $f\in L^{1}\cap L^{2}(\mathbb{R}_{n},dV(\underline{x}))$
is
\[
\mathcal{F}\left[f\right](\underline{\xi})=\widehat{f}(\underline{\xi})=\frac{1}{(2\pi)^{\frac{n}{2}}}\int_{\mathbb{R}^{n}}e^{-i<\underline{x},\underline{\xi}>}f(\underline{x})dV(\underline{x}).
\]
It is an invertible transform and its inverse is
\[
\mathcal{F}^{-1}\left[f\right](\underline{x})=\frac{1}{(2\pi)^{\frac{n}{2}}}\int_{\mathbb{R}^{n}}e^{i<\underline{x},\underline{\xi}>}\widehat{f}(\underline{\xi})dV(\underline{\xi}).
\]
In the sequel, we shall use the two operators
$$
A_{k}f(\underline{x})=x_{k}f(\underline{x})\;\;\hbox{and}\;\;
B_{k}f(\underline{x})=\partial_{x_{k}}f(\underline{x}),\;k=1,2,\cdots,n.
$$
Using Theorem \ref{uncer cas general}, we obtain
\[
\left\Vert A_{k}f\right\Vert _{L^{2}(\mathbb{R}_{n},dV(\underline{x}))}\left\Vert B_{k}f\right\Vert _{L^{2}(\mathbb{R}_{n},dV(\underline{x}))}\geq\frac{1}{2}\left\vert <\left[A_{k},B_{k}\right]f,f>\right\vert,
\]
which reads otherwise as
\begin{equation}\label{uncer clifford-fourier}
\|x_{k}f\|_{L^{2}(\mathbb{R}_{n},dV(\underline{x}))}\|\xi_{k}\widehat{f}\|_{L^{2}(\mathbb{R}_{n},dV(\underline{x}))}\geq\frac{1}{2}\|f\|_{L^{2}(\mathbb{R}_{n},dV(\underline{x}))}^{2}.
\end{equation}
\section{Clifford-Wavelet Transform}
In this section we introduce the concept of the Clifford-wavelet transform and some of its important properties to be used later. In this context, a function $\psi\in L^{1}\cap L^{2}(\mathbb{R}^{n},dV(\underline{x}))$ will be considered as a Clifford mother wavelet. To join the admissibility assumptions in
the real case, here-also we assume that
\begin{itemize}
	\item $\widehat{\psi}(\underline{\xi})\left[\widehat{\psi}(\underline{\xi})\right]^{\dagger}$ is scalar.
	\item $\mathcal{A}_{\psi}=\displaystyle(2\pi)^{n}\int_{\mathbb{R}^{n}}\frac{\widehat{\psi}(\underline{\xi})\left[\widehat{\psi}(\underline{\xi})\right]^{\dagger}}{|\underline{\xi}|^{n}}dV(\underline{\xi})<\infty.$
\end{itemize}
For $(a,\underline{b},s)\in\mathbb{R}^{+}\times\mathbb{R}^{n}\times Spin(n)$,
we denote
\[
\psi^{a,\underline{b},s}(\underline{x})=\frac{1}{a^{\frac{n}{2}}}s\psi(\frac{\overline{s}(\underline{x}-\underline{b})s}{a})\overline{s}.
\]
It holds in fact that these copies are also admissible and that
\[
\mathcal{A}_{\psi^{a,\underline{b},s}}=\frac{a^{n/2}}{(2\pi)^{n}}\mathcal{A}_{\psi}<\infty.
\]
\begin{prop}
	The set $\left\{\psi^{a,\underline{b},s},a>0,\underline{b}\in\mathbb{R}^{n},\,s\in Spin(n)\right\}$ is dense in the space $L^{2}(\mathbb{R}^{n},dV(\underline{x}))$.
\end{prop}
\begin{proof}
	Let $f$ be an analyzed function such that
	\[
	<\psi^{a,\underline{b},s},f>_{L^{2}(\mathbb{R}^{n},dV(\underline{x}))}=0,\;\forall a>0,\,\underline{b}\in\mathbb{R}^{n}\text{ and }s\in Spin(n).
	\]
	We shall prove that $f=0$. Due to the Parseval identity of the Clifford-Fourier
	transform, we obtain
	\[
	<\psi^{a,\underline{b},s},f>_{L^{2}(\mathbb{R}^{n},dV(\underline{x}))}=<\widehat{\psi^{a,\underline{b},s}},\widehat{f}>_{L^{2}(\mathbb{R}^{n},dV(\underline{x}))}=0.
	\]
	Since,
	\[
	<\widehat{\psi^{a,\underline{b},s}},\widehat{f}>_{L^{2}(\mathbb{R}^{n},dV(\underline{x}))}=a^{\frac{n}{2}}\int_{\mathbb{R}^{n}}e^{i<\underline{b},\underline{\xi}>}\left[\widehat{\psi}(a\overline{s}\underline{\xi}s)\right]^{\dagger}\widehat{f}(\underline{\xi})dV(\underline{\xi}),
	\]
	then,
	\[
	s\left[\widehat{\psi}(a\overline{s}\underline{\xi}s)\right]^{\dagger}s\widehat{f}(\underline{\xi})=0.
	\]
	Recall now that for a fixed $\underline{\xi}\not=0$ in $\mathbb{R}^{n}$,
	\[
	\left\{ a\overline{s}\underline{\xi}s,a>0\text{ and }s\in Spin(n)\right\} =\mathbb{R}^{n}.
	\]
	It results that
	\[
	\widehat{f}=0\quad and\quad f\equiv0.
	\]
\end{proof}
\begin{defn} (Clifford Wavelet Transform) The Clifford-wavelet transform of an analyzed function $f$ with respect to the mother wavelet $\psi$ is
	\[
	T_{\psi}\left[f\right](a,\underline{b},s)=\int_{\mathbb{R}^{n}}\left[\psi^{a,\underline{b},s}(\underline{x})\right]^{\dagger}f(\underline{x})dV(\underline{x}).
	\]
\end{defn}
\begin{defn}
	(Inner product relation) Let $\mathcal{H}_{\psi}=\left\{ T_{\psi}\left[f\right],\;f\in L^{2}(\mathbb{R}^{n},dV(\underline{x}))\right\} $
	be the image of $L^{2}(\mathbb{R}_{n},dV(\underline{x}))$ relatively to the operator $T_{\psi}$. We define the inner product by
	\[
	\left[T_{\psi}\left[f\right],T_{\psi}\left[g\right]\right]=\frac{1}{\mathcal{A}_{\psi}}\int\limits _{Spin(n)}\int\limits _{\mathbb{R}^{n}}\int\limits _{\mathbb{R}^{+}}(T_{\psi}\left[f\right](a,\underline{b},s))^{\dagger}T_{\psi}\left[g\right](a,\underline{b},s)\frac{da}{a^{n+1}}dV(\underline{b})ds,
	\]
	where $ds$ stands for the Haar measure on $Spin(n)$.
\end{defn}
\begin{prop}
	$T_{\psi}:L^{2}(\mathbb{R}^{n},dV(\underline{x}))\longrightarrow\mathcal{H}_{\psi}$
	is an isometry.
\end{prop}
\begin{proof}
	We have to show that
	\begin{equation}
	\left[T_{\psi}\left[f\right],T_{\psi}\left[g\right]\right]=<f,g>_{L^{2}(\mathbb{R}^{n},dV(\underline{x}))}.\label{iso clifford ond}
	\end{equation}
	Put
	\[
	\begin{cases}
	\Phi_{\psi}(a,s,\underline{\xi})\left[f\right](-\underline{b})
	&\left[\left[\widehat{\psi}(a\overline{s}\underline{\xi}s)\right]^{\dagger}\overline{s}\widehat{f}(\underline{\xi})\right](-\underline{b})\\
	\Phi_{\psi}(a,s,\underline{\xi})\left[g\right](-\underline{b})&=\left[\left[\widehat{\psi}(a\overline{s}\underline{\xi}s)\right]^{\dagger}\overline{s}\widehat{g}(\underline{\xi})\right](-\underline{b}).
	\end{cases}
	\]
	Hence,
	\[
	\begin{cases}
	T_{\psi}\left[f\right](a,\underline{b},s) & =a^{\frac{n}{2}}s(2\pi)^{\frac{n}{2}}\widehat{\Phi_{\psi}(a,\underline{\xi},s)}\left[f\right](-\underline{b}),\\
	\\
	T_{\psi}\left[g\right](a,\underline{b},s) & ==a^{\frac{n}{2}}s(2\pi)^{\frac{n}{2}}\widehat{\Phi_{\psi}(a,\underline{\xi},s)}\left[g\right](-\underline{b}).
	\end{cases}
	\]
	Applying Parseval formula, we get
	\[
	\left\langle \widehat{\Phi_{\psi}(a,\bullet,s)\left[f\right]},\widehat{\Phi_{\psi}(a,\bullet,s)\left[g\right]}\right\rangle =\left\langle \Phi_{\psi}(a,\bullet,s)\left[f\right],\Phi_{\psi}(a,\bullet,s)\left[g\right]\right\rangle .
	\]
	Next,
	\begin{align*}
		&\left[T_{\psi}\left[f\right],T_{\psi}\left[g\right]\right]\\ &=\frac{1}{(2\pi)^{n}\mathcal{A}_{\psi}}\intop_{Spin(n)}\int\limits_{\mathbb{R}^{+}}\left\{\int_{\mathbb{R}^{n}}(\Phi_{\psi}(a,\underline{\xi},s)\left[f\right](\underline{\xi}))^{\dagger}\Phi_{\psi}(a,\underline{\xi},s)\left[g\right](\underline{\xi})dV(\underline{b})\right\} \frac{da}{a}ds\\
		&=\frac{1}{(2\pi)^{n}\mathcal{A}_{\psi}}\intop_{Spin(n)}\int\limits_{\mathbb{R}^{+}}\left\{\int_{\mathbb{R}^{n}}\left[(\left[\widehat{\psi}(a\overline{s}\underline{\xi}s)\right]^{\dagger}\overline{s}\widehat{f}(\underline{\xi})\right]^{\dagger}\left[\widehat{\psi}(a\overline{s}\underline{\xi}s)\right]^{\dagger}\overline{s}\widehat{g}(\underline{\xi})dV(\underline{\xi})\right\} \frac{da}{a}ds\\
		&=\frac{1}{(2\pi)^{n}\mathcal{A}_{\psi}}\int_{Spin(n)}\int\limits_{\mathbb{R}^{+}}\left\{\int_{\mathbb{R}^{n}}\left[\widehat{f}(\underline{\xi})\right]^{\dagger}s\widehat{\psi}(a\overline{s}\underline{\xi}s)\left[\widehat{\psi}(a\overline{s}\underline{\xi}s)\right]^{\dagger}\overline{s}\widehat{g}(\underline{\xi})dV(\underline{\xi})\right\} \frac{da}{a}ds\\
		&=\frac{1}{(2\pi)^{n}\mathcal{A}_{\psi}}\int_{Spin(n)}\int\limits_{\mathbb{R}^{+}}\left\{\int_{\mathbb{R}^{n}}\left[\widehat{f}(\underline{\xi})\right]^{\dagger}s\widehat{\psi}(a\overline{s}\underline{\xi}s)\left[\widehat{\psi}(a\overline{s}\underline{\xi}s)\right]^{\dagger}\overline{s}\widehat{g}(\underline{\xi})dV(\underline{\xi})\right\} \frac{da}{a}ds\\
		&=\frac{1}{(2\pi)^{n}\mathcal{A}_{\psi}}\int_{\mathbb{R}^{n}}\left[\widehat{f}(\underline{\xi})\right]^{\dagger}\left\{ \int_{Spin(n)}\int\limits _{\mathbb{R}^{+}}s\widehat{\psi}(a\overline{s}\underline{\xi}s)\left[\widehat{\psi}(a\overline{s}\underline{\xi}s)\right]^{\dagger}\overline{s}\frac{da}{a}ds\right\} \widehat{g}(\underline{\xi})dV(\underline{\xi}).
	\end{align*}
	Observing now that
	\begin{equation}\label{eq:Cpsi shit}
	\int_{Spin(n)}\int\limits_{\mathbb{R}^{+}}s\widehat{\psi}(a\overline{s}\underline{\xi}s)\left[\widehat{\psi}(a\overline{s}\underline{\xi}s)\right]^{\dagger}\overline{s}\frac{da}{a}ds=\frac{\mathcal{A}_{\psi}}{(2\pi)^{n}},
	\end{equation}
	we get immediately
	\[
	\left[T_{\psi}\left[f\right],T_{\psi}\left[g\right]\right]
	=\frac{1}{(2\pi)^{n}\mathcal{A}_{\psi}}\int_{\mathbb{R}^{n}}\left[\widehat{f}(\underline{\xi})\right]^{\dagger}\left\{\int_{\mathcal{S}^{n-1}}\int\limits _{\mathbb{R}^{+}}\Gamma(t,\nu)dS(\underline{\nu})\frac{da}{t}\right\} \widehat{g}(\underline{\xi})dV(\underline{\xi}),
	\]
	where we denoted $\Gamma(t,\nu)=\widehat{\psi}(t\underline{\nu})\left[\widehat{\psi}(t\underline{\nu})\right]^{\dagger}$. Otherwise, by taking $\underline{u}=t\underline{\nu}$, we obtain
	\begin{align}
		\left[T_{\psi}\left[f\right],T_{\psi}\left[g\right]\right] &=\frac{1}{(2\pi)^{n}\mathcal{A}_{\psi}}\int_{\mathbb{R}^{n}}\left[\widehat{f}(\underline{\xi})\right]^{\dagger}\left\{\int_{\mathbb{R}^{n}}\frac{\widehat{\psi}(\underline{u})\left[\widehat{\psi}(\underline{u})\right]^{\dagger}}{\left\vert \underline{u}\right\vert ^{n}}dV(\underline{u})\right\} \widehat{g}(\underline{\xi})dV(\underline{\xi})\nonumber \\
		&=\int_{\mathbb{R}^{n}}\left[\widehat{f}(\underline{\xi})\right]^{\dagger}\widehat{g}(\underline{\xi})dV(\underline{\xi})\nonumber \\
		&=<\widehat{f},\widehat{g}>\nonumber \\
		&=<f,g>.\label{clifford WT isom}
	\end{align}
	As a result we have
	\begin{equation}
	\int\limits _{Spin(n)}\int\limits _{\mathbb{R}^{n}}\int\limits _{\mathbb{R}^{+}}(T_{\psi}\left[f\right](a,\underline{b},s))^{2}\frac{da}{a^{n+1}}dV(\underline{b})ds=\mathcal{A}_{\psi}\left\Vert f\right\Vert _{2}^{2}\label{eq:norme equality}
	\end{equation}
\end{proof}
As a result of the last Proposition and as in the real case, we have
here a Clifford-wavelet reconstruction formula.
\begin{prop}
	For all $f\in L^{2}(\mathbb{R}^{n},dV(\underline{x}))$ we have
	\[
	f(\underline{x})=\frac{1}{A_{\psi}}\int\limits _{Spin(n)}\int\limits _{\mathbb{R}^{n}}\int\limits _{\mathbb{R}^{+}}\psi^{a,\underline{b},s}(\underline{x})T_{\psi}\left[f\right](a,\underline{b},s)\frac{da}{a^{n+1}}dV(\underline{b})ds.
	\]
\end{prop}
\section{Clifford wavelet uncertainty principle}
In this section, we established the Heisenberg uncertainty principle for the Clifford wavelet transform. Backgrounds may be found in \cite{Bahri2011b}.  \begin{thm}\label{MainTheorem}
	Let $\psi\in L^{2}(\mathbb{R}^{n},dV(\underline{x}))$ an admissible
	Clifford mother wavelet. Then for $f\in L^{2}(\mathbb{R}^{n},dV(\underline{x}))$
	the following inequality holds
	\begin{equation}
	\displaystyle\left(\displaystyle\int_{Spin(n)}\displaystyle\int_{\mathbb{R}^{+}}\left\Vert b_{k}T\left[f\right](a,\cdot,s)\right\Vert _{2}^{2}\frac{da}{a^{n+1}}ds\right)^{\frac{1}{2}}\left\Vert \xi_{k}\widehat{f}\right\Vert _{2}\geq\frac{(2\pi)^{\frac{n}{2}}}{2}\sqrt{A_{\psi}}\left\Vert f\right\Vert _{2}^{2},
	\end{equation}
	where $k=1,2,\cdots,n$.
\end{thm}
To prove this result we need the following lemma :
\begin{lem}\label{MainResultLemma}
	\[
	\displaystyle\intop_{Spin(n)}\displaystyle\int_{\mathbb{R}^{+}}\displaystyle\int_{\mathbb{R}^{n}}\left\{\xi_{k}\widehat{T\left[f\right]}(a,\underline{\xi},s)\right\}^{2}dV(\underline{\xi})\frac{da}{a^{n+1}}ds=\frac{A_{\psi}}{(2\pi)^{n}}\left\Vert \xi_{k}\widehat{f}\right\Vert _{2}^{2}.
	\]
\end{lem}
\begin{proof}
	As $\widehat{\psi^{a,\underline{b},s}}(\underline{\xi})=a^{\frac{n}{2}}e^{-i<\underline{b},\underline{\xi}>}s\widehat{\psi}(a\overline{s}\underline{\xi}s)\overline{s}$,
	we get
	\[
	T\left[f\right](a,\underline{b},s)=a^{\frac{n}{2}}\mathcal{F}^{-1}\left[\overline{s}\left[\widehat{\psi}(a\overline{s}\cdot s)\right]^{\dagger}s\widehat{f}(\cdot)\right](\underline{b})
	\]
	and
	\begin{equation}
	\widehat{T\left[f\right]}(a,\underline{\xi},s)=a^{\frac{n}{2}}\overline{s}\left[\widehat{\psi}(a\overline{s}\underline{\xi}s)\right]^{\dagger}s\widehat{f}(\underline{\xi}).\label{fourier d'ondelette clifford}
	\end{equation}
	Therefore,
	\begin{align}
		&\displaystyle\intop_{\mathbb{R}^{n}}\left\{\xi_{k}\widehat{T\left[f\right]}(a,\underline{\xi},s)\right\} ^{2}dV(\underline{\xi})\\ &=\intop_{\mathbb{R}^{n}}\left\{\xi_{k}a^{\frac{n}{2}}\overline{s}\left[\widehat{\psi}(a\overline{s}\underline{\xi}s)\right]^{\dagger}s\widehat{f}(\underline{\xi})\right\} ^{2}dV(\underline{\xi})\nonumber \\		 &=\intop_{\mathbb{R}^{n}}\xi_{k}a^{\frac{n}{2}}\overline{s}\left[\widehat{\psi}(a\overline{s}\underline{\xi}s)\right]^{\dagger}s\widehat{f}(\underline{\xi})\left\{\xi_{k}a^{\frac{n}{2}}\overline{s}\left[\widehat{\psi}(a\overline{s}\underline{\xi}s)\right]^{\dagger}s\widehat{f}(\underline{\xi})\right\} ^{\dagger}dV(\underline{\xi})\nonumber \\
		&=\intop_{\mathbb{R}^{n}}\xi_{k}^{2}a^{n}\overline{s}\left[\widehat{\psi}(a\overline{s}\underline{\xi}s)\right]^{\dagger}s\widehat{f}(\underline{\xi})\left\{ \overline{s}\left[\widehat{\psi}(a\overline{s}\underline{\xi}s)\right]^{\dagger}s\widehat{f}(\underline{\xi})\right\} ^{\dagger}dV(\underline{\xi})\nonumber \\
		&=\intop_{\mathbb{R}^{n}}\xi_{k}^{2}a^{n}\overline{s}\left[\widehat{\psi}(a\overline{s}\underline{\xi}s)\right]^{\dagger}s\widehat{f}(\underline{\xi})\left[\widehat{f}(\underline{\xi})\right]^{\dagger}\overline{s}\widehat{\psi}(a\overline{s}\underline{\xi}s)sdV(\underline{\xi})\nonumber \\
		&=\intop_{\mathbb{R}^{n}}\xi_{k}^{2}a^{n}\left[\widehat{\psi}(a\overline{s}\underline{\xi}s)\right]^{\dagger}\widehat{f}(\underline{\xi})\left[\widehat{f}(\underline{\xi})\right]^{\dagger}\widehat{\psi}(a\overline{s}\underline{\xi}s)dV(\underline{\xi})\nonumber \\
		&=\intop_{\mathbb{R}^{n}}\xi_{k}^{2}a^{n}\left\{ \left[\widehat{\psi}(a\overline{s}\underline{\xi}s)\right]^{\dagger}\widehat{\psi}(a\overline{s}\underline{\xi}s)\right\} \left\{ \widehat{f}(\underline{\xi})\left[\widehat{f}(\underline{\xi})\right]^{\dagger}\right\} dV(\underline{\xi}).\label{xi X fourierT(f)}
	\end{align}
	Using of (\ref{xi X fourierT(f)}) we obtain
	\begin{align*}
		&\intop_{Spin(n)}\intop_{\mathbb{R}^{+}}\intop_{\mathbb{R}^{n}}\left\{ \xi_{k}\widehat{T\left[f\right]}(a,\underline{\xi},s)\right\} ^{2}dV(\underline{\xi})\frac{da}{a^{n+1}}ds\\ &=\intop_{Spin(n)}\intop_{\mathbb{R}^{+}}\intop_{\mathbb{R}^{n}}\xi_{k}^{2}a^{n}\left\{\left[\widehat{\psi}(a\overline{s}\underline{\xi}s)\right]^{\dagger}\widehat{\psi}(a\overline{s}\underline{\xi}s)\right\} \left\{\widehat{f}(\underline{\xi})\left[\widehat{f}(\underline{\xi})\right]^{\dagger}\right\} dV(\underline{\xi})\frac{da}{a^{n+1}}ds\\
		&=\intop_{Spin(n)}\intop_{\mathbb{R}^{+}}\intop_{\mathbb{R}^{n}}\xi_{k}^{2}\left\{\left[\widehat{\psi}(a\overline{s}\underline{\xi}s)\right]^{\dagger}\widehat{\psi}(a\overline{s}\underline{\xi}s)\right\}\left\{\widehat{f}(\underline{\xi})\left[\widehat{f}(\underline{\xi})\right]^{\dagger}\right\} dV(\underline{\xi})\frac{da}{a}ds\\
		&=\intop_{\mathbb{R}^{n}}\left\{\intop_{Spin(n)}\intop_{\mathbb{R}^{+}}\frac{\left[\widehat{\psi}(a\overline{s}\underline{\xi}s)\right]^{\dagger}\widehat{\psi}(a\overline{s}\underline{\xi}s)}{a}dads\right\}\xi_{k}^{2}\left\{\widehat{f}(\underline{\xi})\left[\widehat{f}(\underline{\xi})\right]^{\dagger}\right\} dV(\underline{\xi}).
	\end{align*}
	According to (\ref{eq:Cpsi shit}), we get in fin
	\[
	\displaystyle\intop_{Spin(n)}\intop_{\mathbb{R}^{+}}\intop_{\mathbb{R}^{n}}\left\{\xi_{k}\widehat{T\left[f\right]}(a,\underline{\xi},s)\right\}^{2}dV(\underline{\xi})\frac{da}{a^{n+1}}ds=\frac{A_{\psi}}{(2\pi)^{n}}\left\Vert\xi_{k}\widehat{f}\right\Vert _{2}^{2}.
	\]
\end{proof}
\begin{proof} \textit{of Theorem \ref{MainTheorem}.}
	Using (\ref{uncer clifford-fourier}) and setting $\underline{x}=\underline{b}\in\mathbb{R}^{n}$, we deduce that
	\[
	\left\Vert b_{k}T\left[f\right](a,\cdot,s)\right\Vert _{2}\left\Vert \xi_{k}\widehat{T\left[f\right]}(a,\cdot,s)\right\Vert_{2}\geq\frac{1}{2}\left\Vert T\left[f\right](a,\cdot,s)\right\Vert_{2}^{2}.
	\]
	Therefore
	$$
	\int_{Spin(n)}\int_{\mathbb{R}^{+}}\left\Vert b_{k}T\left[f\right](a,\cdot,s)\right\Vert _{2}\left\Vert \xi_{k}\widehat{T\left[f\right]}(a,\cdot,s)\right\Vert _{2}\frac{da}{a^{n+1}}ds
	$$
	$$
	\qquad\qquad\qquad\qquad\qquad\qquad\geq\frac{1}{2}\int_{Spin(n)}\int_{\mathbb{R}^{+}}\left\Vert T\left[f\right](a,\cdot,s)\right\Vert _{2}^{2}\frac{da}{a^{n+1}}ds
	$$
	According to the Cauchy-Schwartz inequality (\ref{cauchy schwartz}),
	it follows that
	\begin{equation*}
		\int_{Spin(n)\times\mathbb{R}^{+}}\left\Vert b_{k}T\left[f\right](a,\cdot,s)\right\Vert _{2}^{2}\frac{da}{a^{n+1}}ds\times\int_{Spin(n)\times\mathbb{R}^{+}}\left\Vert \xi_{k}\widehat{T\left[f\right]}(a,\cdot,s)\right\Vert _{2}^{2}\frac{da}{a^{n+1}}ds
	\end{equation*}
	\begin{equation}\label{uncer trop complique}
	\geq\left(\frac{1}{2}\int_{Spin(n)\times\mathbb{R}^{+}\times\mathbb{R}^{n}}\left\{ T\left[f\right](a,\underline{b},s)\right\} ^{2}dV(\underline{b})\frac{da}{a^{n+1}}ds\right)^{2}.
	\end{equation}
	Now, using Lemma \ref{MainResultLemma} and the fact that the wavelet-transform is an isometry, we get by (\ref{eq:norme equality})
	\[
	\int\limits _{Spin(n)}\int\limits _{\mathbb{R}^{n}}\int\limits _{\mathbb{R}^{+}}(T_{\psi}\left[f\right](a,\underline{b},s))^{2}\frac{da}{a^{n+1}}dV(\underline{b})ds=\mathcal{A}_{\psi}\left\Vert f\right\Vert _{2}^{2}.
	\]
	The inequality (\ref{uncer trop complique}) becomes
	\begin{equation}
	\left(\intop_{Spin(n)}\intop_{\mathbb{R}^{+}}\left\Vert b_{k}T\left[f\right](a,\cdot,s)\right\Vert _{2}^{2}\frac{da}{a^{n+1}}ds\right)^{\frac{1}{2}}\times\left(\frac{A_{\psi}}{(2\pi)^{n}}\left\Vert \xi_{k}\widehat{f}\right\Vert _{2}^{2}\right)^{\frac{1}{2}}\geq\frac{1}{2}\mathcal{A}_{\psi}\left\Vert f\right\Vert _{2}^{2}.\label{presque uncer}
	\end{equation}
	Hence, we obtain
	\[
	\left(\intop_{Spin(n)}\intop_{\mathbb{R}^{+}}\left\Vert b_{k}T\left[f\right](a,\cdot,s)\right\Vert _{2}^{2}\frac{da}{a^{n+1}}ds\right)^{\frac{1}{2}}\left\Vert \xi_{k}\widehat{f}\right\Vert _{2}\geq\frac{(2\pi)^{\frac{n}{2}}}{2}\sqrt{A_{\psi}}\left\Vert f\right\Vert _{2}^{2}.
	\]
\end{proof}
\section{Conclusion}
In this paper, a new uncertainty principle associated with the continuous
wavelet transform in the Clifford algebra's settings has been formulated
and proved. Starting from the definition of real Clifford algebra
and the real continuous wavelet transform, we defined a continuous Clifford-Wavelet
Transform, presented its proprieties and formulated an uncertainty
relation based on the uncertainty principle for the Clifford-Fourier
transform.

\end{document}